\documentclass[12pt]{article}

\usepackage{amsmath,amssymb,amsfonts}
\usepackage{algorithmic}
\usepackage{graphicx}
\usepackage{textcomp}
\usepackage{algorithm}
\usepackage{subcaption}
\usepackage{bm}
\usepackage{amsthm}
\usepackage{array}
\usepackage{booktabs}
\usepackage{url}
\usepackage{float}
\usepackage{placeins}
\usepackage{geometry}
\usepackage{hyperref}
\usepackage{adjustbox}

\newtheorem{theorem}{Theorem}
\newtheorem{lemma}{Lemma}
\newtheorem{proposition}{Proposition}

\title{Re-Rooting-Based Fault-Tolerant Broadcasting in Dense Gaussian Networks}

\author{Bader Albader, Mohamed R. Al-Mulla, and Galal Hassan\\
\small Department of Computer Science, Kuwait University, Kuwait\\
\small \texttt{albader@cs.ku.edu.kw}}

\date{}

\begin{document}

\maketitle

\begin{abstract}
Dense Gaussian networks provide degree-4 interconnection topologies with small diameter and regular structure, making them suitable for efficient one-to-all broadcasting. However, node failures can disrupt the broadcast process when faulty nodes occupy internal forwarding positions. This paper proposes a lightweight fault-tolerant broadcasting method based on dynamic source relocation, or re-rooting. Instead of constructing redundant spanning trees or backup routing structures, the proposed method selects a new source node so that the faulty nodes are located at graph distance \(k\), the network diameter, from the new source. Consequently, faulty nodes become leaf-level nodes in the broadcast process and are not required to forward the message. For the single-fault case, the new source is selected directly from the graph-distance-\(k\) boundary of the faulty node. For the two-fault case, we prove that for any pair of faulty nodes in \(G(k+(k+1)i)\), there exists a node whose graph distance from both faulty nodes is exactly \(k\). The source-selection procedure requires \(O(k)\) time. Since the original one-to-all broadcast completes in \(k\) parallel steps and the relocation distance is at most \(k\), the proposed method completes in at most \(2k\) steps in the worst case. We also show that the two-fault guarantee does not generally extend to arbitrary three-fault configurations by giving a counterexample in \(G(3+4i)\). Simulation results confirm complete delivery to all non-faulty nodes under the tested one- and two-node failure scenarios, while the baseline broadcast may fail when faulty nodes occur at internal forwarding positions.
\end{abstract}

\noindent\textbf{Keywords:} Dense Gaussian networks, fault-tolerant broadcasting, re-rooting, source relocation, interconnection networks, one-to-all broadcasting, node failures, parallel communication.

\section{Introduction}

Interconnection topology plays a fundamental role in the performance, scalability, and reliability of parallel and multicore systems. As modern computing systems continue to integrate a large number of processing elements on a single chip, efficient and robust communication mechanisms are becoming increasingly important. In such systems, collective communication operations, particularly broadcasting, are essential for synchronization, data dissemination, and control signaling~\cite{ref6,ref7,ref8,ref10,ref11}.

Dense Gaussian networks have been proposed as an attractive class of interconnection topologies owing to their regular structure, symmetry, and favorable distance properties compared with other topologies, such as tori. Their close relationship with dense circulant graphs enables efficient routing and communication algorithms, making them suitable for high-performance parallel systems~\cite{ref1,ref2,ref3}. Prior work established optimal one-to-all broadcasting mechanisms for dense Gaussian networks under fault-free conditions~\cite{ref2}.

From a practical perspective, dense Gaussian networks are relevant to communication fabrics in Network-on-Chip (NoC) systems, chip multiprocessors, distributed accelerators, and edge-computing architectures, where scalable and low-latency communication is required among many processing elements. In such systems, one-to-all broadcast is commonly used for distributing control signals, synchronization messages, configuration data, cache-coherence information, and workload-management commands. Since these platforms often operate under strict constraints on area, energy, and routing complexity, lightweight fault-tolerant mechanisms are preferable to approaches that require multiple redundant routing structures or complex runtime recovery logic.

However, in practical systems, node failures are inevitable and can significantly degrade communication performance. In particular, even a small number of faulty nodes can disrupt the broadcast process and prevent complete message delivery. Therefore, designing fault-tolerant broadcasting mechanisms that maintain efficiency while ensuring robustness is critical. Recent work on dense Gaussian on-chip networks has considered protection routing through completely independent spanning trees~\cite{ref9}. More broadly, independent and edge-independent spanning tree constructions remain active research topics for reliable communication in interconnection networks~\cite{ref12,ref13,ref14}.

In this paper, we revisit fault-tolerant broadcasting in dense Gaussian networks and propose an efficient approach based on \emph{dynamic source relocation} (also referred to as re-rooting). Instead of relying on additional fault-tolerant communication structures, the proposed method identifies a new source node whose graph distance from each faulty node is equal to the network diameter \(k\), and then performs standard broadcasting from that node. As a result, the method preserves the simplicity of the original broadcast algorithm while effectively neutralizing the effects of faulty nodes.

Unlike fault-tolerant approaches based on independent or completely independent spanning trees, which require the construction and maintenance of multiple disjoint communication structures~\cite{ref9,ref12,ref13,ref14}, the proposed method avoids preprocessing and structural redundancy. Instead, it dynamically relocates the source node to ensure that the faulty nodes lie at the periphery of the broadcast process. This results in a simpler topology-specific mechanism while still achieving complete broadcast delivery for all tested one- and two-node failure scenarios.

The proposed approach was developed for both single- and two-node failure scenarios. The source-selection process is efficient and remains linear in the network diameter \(k\). To evaluate the effectiveness and scalability of the approach, we conducted extensive simulations on dense Gaussian networks with up to \(80401\) nodes under multiple fault-placement modes. The results demonstrate that the proposed method achieves a \textbf{100\% broadcast success rate} in all tested configurations, whereas the baseline broadcast algorithm frequently fails as the network size increases.

The main contributions of this paper are summarized as follows:
\begin{itemize}
	\item We propose a lightweight re-rooting-based fault-tolerant broadcasting method for dense Gaussian networks. Instead of constructing redundant spanning trees, backup paths, or additional routing structures, the proposed method dynamically relocates the effective broadcast source.
	
	\item We develop source-selection algorithms for single-node and two-node failure scenarios. The selected new source \(NS\) places the faulty nodes at graph distance \(k\), the network diameter, so that they become leaf-level nodes in the broadcast process.
	
	\item We prove that for any pair of faulty nodes in \(G(k+(k+1)i)\), there always exists a node whose graph distance from both faulty nodes is exactly \(k\). This provides a deterministic existence guarantee for the two-fault re-rooting strategy.

	\item We show that this guarantee does not generally extend to arbitrary three-node failures by presenting a counterexample in \(G(3+4i)\). This establishes the theoretical boundary of the proposed method.
	
	\item We analyze the computational and communication cost of the method. The source-selection procedure requires \(O(k)\)~time, the original one-to-all broadcast completes in \(k\) parallel steps, and the re-rooted broadcast completes in at most \(2k\) steps in the worst case.
	
	\item We evaluate the proposed method through simulations under one- and two-node failure scenarios and compare it with the original non-fault-tolerant broadcast. The results confirm complete delivery to all non-faulty nodes in the tested configurations.
\end{itemize}

The remainder of this paper is organized as follows. Section~II reviews related work. Section~III introduces dense Gaussian networks. Section~IV reviews the baseline one-to-all broadcasting mechanism. Section~V presents the proposed fault-tolerant approach. Section~VI presents the experimental results, and Section~VII concludes the paper.

\section{Related Work}

Fault-tolerant communication has long been an important research topic in interconnection networks and parallel systems. Broadcasting, in particular, plays a central role in collective communication operations, synchronization, and control signaling. Consequently, considerable effort has been devoted to developing reliable broadcasting mechanisms that can tolerate node and link failures while preserving efficient communication performance.

Classical approaches to fault-tolerant broadcasting often rely on redundant communication structures, such as multiple spanning trees or disjoint routing paths. Independent spanning trees have been widely studied because they provide multiple internally node-disjoint communication paths between nodes, allowing communication to continue even in the presence of failures. Cheng \emph{et al.}~\cite{ref12} provide a recent survey of independent spanning trees in networks and discuss their role in reliable communication and fault-tolerant broadcasting. Recent studies have also developed independent or edge-independent spanning-tree constructions for specific interconnection networks, including bubble-sort networks~\cite{ref13} and folded crossed cubes~\cite{ref14}.

Fault tolerance has also been extensively investigated in the context of Network-on-Chip (NoC) architectures and modern multicore communication systems. Of particular relevance to this work, Pai \emph{et al.}~\cite{ref9} configured protection routing via completely independent spanning trees in dense Gaussian on-chip networks. Their work uses the dense Gaussian topology to construct redundant protection structures for reliable routing. Recent NoC studies have also considered fault-tolerant and quality-of-service-aware routing~\cite{ref15}, adaptive and passage-based routing around faulty regions in 3D mesh NoCs~\cite{ref16}, machine-learning-based NoC optimization~\cite{ref17}, and transformer-based reinforcement learning for fault-tolerant NoC application mapping~\cite{ref18}. These studies illustrate the continued importance of resilient communication mechanisms in many-core and on-chip systems. Pasricha and Dutt~\cite{ref10} and Flich and Bertozzi~\cite{ref11} further emphasize the importance of scalable and reliable communication structures for many-core systems.

Dense Gaussian networks have attracted significant attention because of their symmetry, regularity, and favorable distance properties. Mart\'inez \emph{et al.} \cite{ref2} demonstrated that dense Gaussian networks are suitable topologies for on-chip multiprocessors. Beivide \emph{et al.} \cite{ref3,ref4} investigated Gaussian interconnection networks and their relationship to circulant graphs, whereas Flahive and Bose \cite{ref1} studied their topological properties and communication efficiency.

Existing fault-tolerant communication approaches typically require additional routing structures, redundant paths, adaptive path selection, runtime monitoring, or specialized recovery mechanisms. In contrast, the method proposed in this paper follows a different philosophy. Rather than constructing independent spanning trees, completely independent spanning trees, backup routing paths, or continuously adapting routes during propagation, fault tolerance is achieved through dynamic source relocation. The broadcast source is relocated such that the faulty nodes become leaf-level nodes with respect to the broadcast process, thereby preserving the correctness of the original one-to-all broadcasting algorithm without introducing redundant broadcast structures.

To the best of our knowledge, dynamic source relocation has not previously been investigated as a fault-tolerant broadcasting mechanism for dense Gaussian networks.

\section{Background}

In this section, we introduce the definitions and terminology of dense Gaussian networks and describe their relationship with dense circulant graphs of degree four.

\subsection{Dense Circulant Graphs}

A circulant graph with $N$ vertices $\{0,1,\ldots,N-1\}$ and $m$ jumps $\{j_1,j_2,\ldots,j_m\}$ is defined as an undirected graph in which each vertex $n$, $0\le n\le N-1$, is adjacent to the vertices $n \pm j_i \bmod N$, for $1\le i\le m$. This graph is denoted by $C_N(j_1,j_2,\ldots,j_m)$. The vertex symmetry of circulant graphs allows their analysis to be carried out from any node, typically node zero, which simplifies their study.

As a special case, when $m=2$, we obtain $C_N(j_1,j_2)$ graphs, which are of degree four. In such graphs, there are \(4d\) nodes at graph distance \(d\) from any given node owing to symmetry. For a given diameter $k$, the maximum number of nodes in a $C_N(j_1,j_2)$ graph is given by~\cite{ref3,ref4}
\begin{equation}
	N = 2k^2 + 2k + 1 = k^2 + (k+1)^2.
\end{equation}

To achieve this maximum, we select $j_1 = k$ and $j_2 = k+1$, resulting in a family of graphs known as \emph{dense circulant graphs}. These graphs exhibit minimal diameter among all degree-four circulant graphs~\cite{ref2,ref5,ref7}, making them attractive for use in communication networks.

\subsection{Dense Gaussian Networks}

Two-dimensional labeling of nodes in degree-four circulant graphs using Gaussian integers was introduced in~\cite{ref3,ref6}. In this representation, each node is associated with a Gaussian integer, leading to the concept of dense Gaussian networks, which provide an equivalent but more intuitive geometric interpretation of the model.

Gaussian integers are defined as
\begin{equation}
	\mathbb{Z}[i] = \{x + yi \mid x, y \in \mathbb{Z}\},
\end{equation}
and form a Euclidean domain with norm
\begin{equation}
	N(x + yi) = x^2 + y^2.
\end{equation}

\begin{figure}[H]
	\centering
	\begin{subfigure}{0.45\linewidth}
		\centering
		\includegraphics[width=\linewidth]{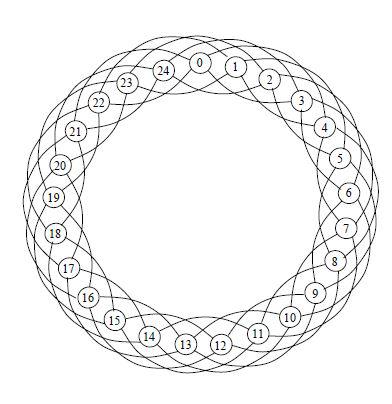}
		\caption{}
	\end{subfigure}
	\hfill
	\begin{subfigure}{0.45\linewidth}
		\centering
		\includegraphics[width=\linewidth]{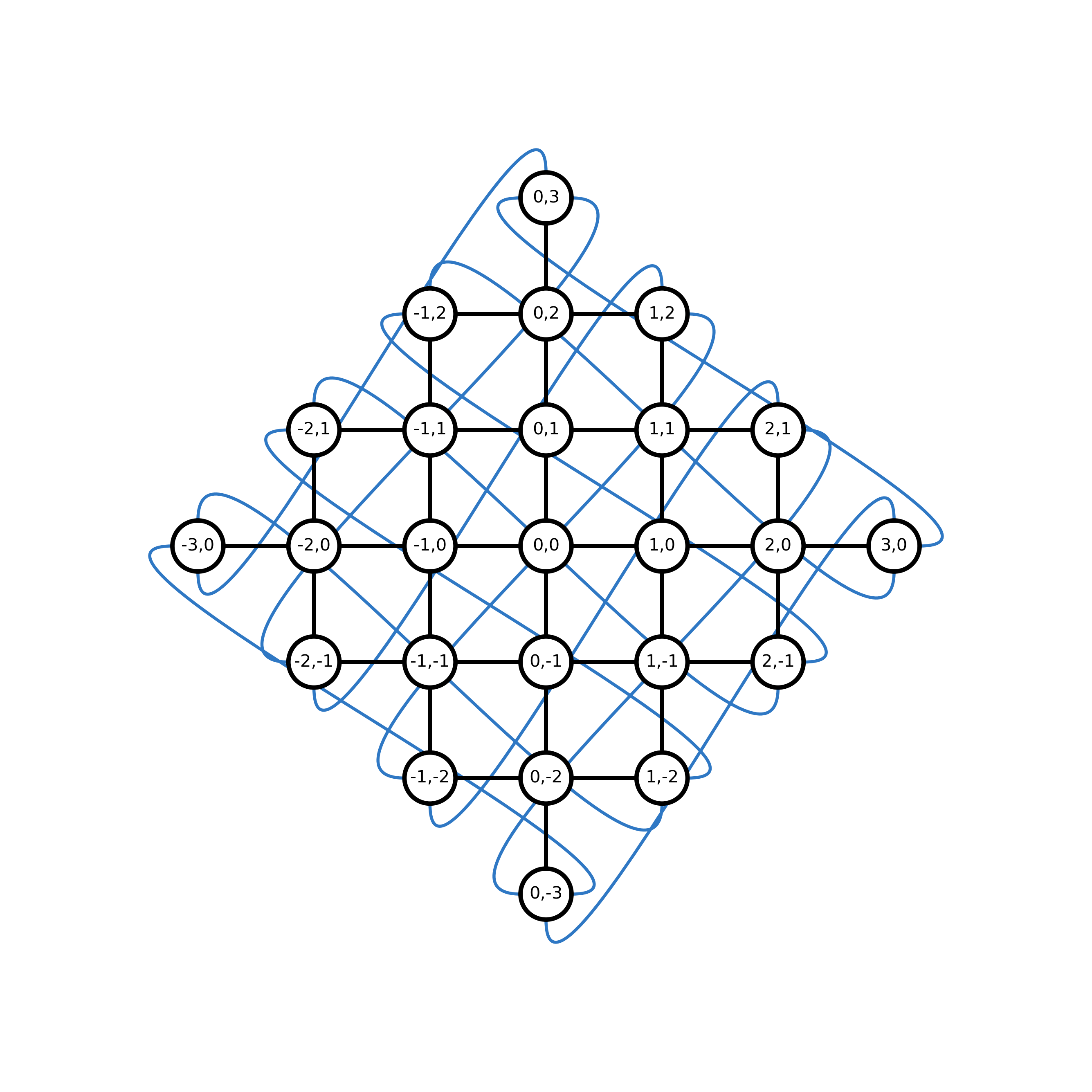}
		\caption{}
	\end{subfigure}
	\caption{(a) $C_{25}(3,4)$ and (b) its isomorphic graph $G(3+4i)$.}
	\label{fig:fig1}
\end{figure}

The Euclidean structure enables division with remainder in $\mathbb{Z}[i]$. For any $\alpha, \pi \in \mathbb{Z}[i]$, $\pi \neq 0$, there exist $q, r \in \mathbb{Z}[i]$ such that $\alpha = q\pi + r$ with $N(r) < N(\pi)$. This allows the definition of a modulo operation on the Gaussian integers.

Given a nonzero $\pi \in \mathbb{Z}[i]$, we define
\begin{equation}
	\mathbb{Z}[i]_{\pi} = \{\alpha \bmod \pi \mid \alpha \in \mathbb{Z}[i]\}.
\end{equation}

For $\pi = a + bi$, the modulo operation can be computed as
\[
\alpha \bmod \pi = \alpha - \left[\frac{\alpha \pi^*}{a^2 + b^2}\right]\pi,
\]
where $\pi^*$ is the conjugate of $\pi$, and $[\cdot]$ denotes rounding applied independently to the real and imaginary parts.

Let $\pi = k + (k+1)i \in \mathbb{Z}[i]$. A dense Gaussian network $G(\pi) = (V, E)$ is defined as
\begin{enumerate}
	\item $V = \mathbb{Z}[i]_{\pi}$,
	\item $E = \{(\alpha,\beta) \in V \times V \mid (\beta - \alpha) \equiv \pm 1, \pm i \pmod{\pi}\}$.
\end{enumerate}

Equivalently, nodes can be represented as integer pairs $(x,y)$ satisfying $|x| + |y| \le k$, where each node $(x,y)$ is connected to $(x \pm 1, y)$ and $(x, y \pm 1)$ under the modulo arithmetic.

It has been shown that dense circulant graphs and dense Gaussian networks are isomorphic~\cite{ref4,ref6}. Fig.~\ref{fig:fig1} illustrates this relationship by showing the dense circulant graph $C_{25}(3,4)$ and its corresponding Gaussian network representation $G(3+4i)$.

For \(N = k^2+(k+1)^2\), the graphs \(C_N(k,k+1)\) and \(G(k+(k+1)i)\) are isomorphic. To move between the Gaussian and integer representations, we define the integer-labeling function
\[
\phi:G(k+(k+1)i)\rightarrow \mathbb{Z}_N
\]
by
\[
\phi(x+yi)\equiv kx+(k+1)y \pmod N.
\]
Thus, each Gaussian-network node \(x+yi\) is assigned a unique integer label modulo \(N\). Conversely, the inverse correspondence maps each integer label in \(\mathbb{Z}_N\) to its associated Gaussian node in \(\mathbb{Z}[i]_{k+(k+1)i}\).

Because this correspondence is an isomorphism, node differences are preserved under the labeling. Hence, for any two nodes \(U,V\in G(k+(k+1)i)\),
\[
\phi(U-V)\equiv \phi(U)-\phi(V) \pmod N.
\]
This property allows relative node positions, wrap-around adjacency, and distance-based source-selection conditions to be expressed using modular integer arithmetic.

Both representations will be used interchangeably. For example, in the case \(k=3\), the node \((-2,1)\) in the Gaussian representation has integer label
\[
\phi(-2+i)\equiv 3(-2)+4(1)\equiv -2\equiv 23 \pmod{25},
\]
and therefore corresponds to node \(23\) in the circulant graph.

The modulo operation also determines the wrap-around adjacency. For example, in $G(3+4i)$, the right neighbor of $(2,1)$ is
\[
(2,1) + (1,0) = (3,1) \equiv (0,-3) \pmod{(3,4)}.
\]
The upper neighbor is
\[
(2,1) + (0,1) = (2,2) \equiv (-1,-2) \pmod{(3,4)}.
\]

Throughout this paper, the distance \(d(u,v)\) between two nodes \(u\) and \(v\) denotes the graph distance in the dense Gaussian network, namely the minimum number of communication links, or hops, on any path connecting \(u\) and \(v\). Since the diameter of \(G(k+(k+1)i)\) is \(k\), every pair of nodes satisfies
\[
d(u,v)\leq k.
\]
Accordingly, a node at graph distance \(k\) from another node is located on the boundary of the network with respect to that node.

\section{One-to-All Broadcasting}

An optimal one-to-all broadcast routing algorithm for dense Gaussian networks was presented in~\cite{ref3,ref4,ref6}. This algorithm is simple, symmetric, and identical for all nodes and packets, making it well suited for efficient hardware implementation.

In this section, we briefly review this algorithm, which serves as the \emph{baseline broadcast process} for the fault-tolerant approach developed in the next section.

Given an arbitrary node as the source, the remaining nodes of the network can be partitioned into four disjoint subsets. These regions can be viewed as four discrete right-angled triangles with legs of size $k$ and $k+1$, each containing $k(k+1)/2$ nodes, as illustrated in Fig.~\ref{fig:fig2}.

We refer to each such region as a \emph{$k$-triangle}. Thus, the entire network can be interpreted as a central source node surrounded by four $k$-triangular quadrants.

Referring to Fig.~\ref{fig:fig2}, node $(0,0)$ has four neighbors: $(1,0)$, $(-1,0)$, $(0,1)$, and $(0,-1)$. Each neighbor lies at the right angle of one of the four $k$-triangular quadrants: southeast (SE), northeast (NE), northwest (NW), and southwest (SW). From each of these nodes, the broadcast propagates outward, reaching $d$ nodes at distance $d$, for $1 < d \le k$.

A router model with full-duplex links and all-port capabilities was assumed. Routers support both unicast and broadcast routing, with the first header bit ($B/U$) indicating the routing mode of the packet. For broadcast packets, a \emph{distance} field is initialized to the network diameter $k$ and is decremented at each hop. The broadcast is terminated when this field reaches zero.

The third header field, denoted as \emph{NSEW}, consists of four bits indicating the output ports (North, South, East, and West) to which the packet will be forwarded. The resulting header size is compact, requiring $\log_2 k$ bits plus five additional control bits.

At the beginning of the broadcast, the source node injects a packet with:
\[
B/U = 1, \quad \textit{distance} = k, \quad \textit{NSEW} = 1111.
\]

\begin{figure}[H]
	\centering
	\includegraphics[width=0.45\linewidth]{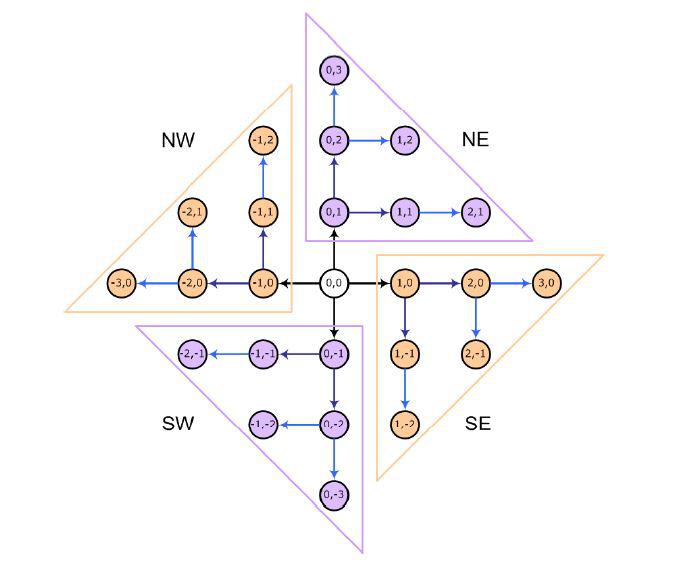}
	\caption{One-to-all Broadcasting Algorithm.}
	\label{fig:fig2}
\end{figure}

\begin{algorithm}[H]
	\caption{Broadcast forwarding algorithm}
	\begin{algorithmic}[1]
		\REQUIRE $B/U$ mask, distance, NSEW: header of the incoming packet
		\STATE $N\_\text{out}\_\text{mask}=1010$;
		\STATE $S\_\text{out}\_\text{mask}=0101$;
		\STATE $E\_\text{out}\_\text{mask}=0110$;
		\STATE $W\_\text{out}\_\text{mask}=1001$;
		\IF{distance $< k$}
		\STATE consume packet
		\ENDIF
		\IF{distance $> 0$}
		\STATE distance $\leftarrow$ distance $-1$
		\IF{NSEW \& 1000}
		\STATE send packet to $N$ with NSEW $\leftarrow$ NSEW $\&$ $N\_\text{out}\_\text{mask}$
		\ENDIF
		\IF{NSEW \& 0100}
		\STATE send packet to $S$ with NSEW $\leftarrow$ NSEW $\&$ $S\_\text{out}\_\text{mask}$
		\ENDIF
		\IF{NSEW \& 0010}
		\STATE send packet to $E$ with NSEW $\leftarrow$ NSEW $\&$ $E\_\text{out}\_\text{mask}$
		\ENDIF
		\IF{NSEW \& 0001}
		\STATE send packet to $W$ with NSEW $\leftarrow$ NSEW $\&$ $W\_\text{out}\_\text{mask}$
		\ENDIF
		\ENDIF
	\end{algorithmic}
	\label{alg:broadcast}
\end{algorithm}

In the first step, the source transmits packets in all four directions, reaching the roots of the four $k$-triangles. Each output port applies a predefined bitmask to update its header. For example, the North output uses the mask $1010$, directing the packet to the NE quadrant.

As the broadcast progresses, nodes located along the axes propagate packets in two directions, whereas other nodes forward packets along a single direction. This mechanism ensures that duplicate packets are not generated.

An important property of this algorithm is that the network utilization is balanced. At step $d$, exactly $d$ packets traverse each quadrant, thereby providing uniform load distribution across all directions.

The algorithm can be efficiently implemented in hardware using bitmasking operations. Each router forwards the packet to the output ports indicated in the NSEW field, and each output port updates this field using a logical AND operation with its corresponding mask.

Because no duplicate transmissions occur, the total number of links used is $N - 1 = 2k^2 + 2k$, which is optimal for one-to-all broadcasting. Furthermore, the total broadcast time is proportional to the network diameter $k$, which scales as $O(\sqrt{N})$.

The broadcast forwarding procedure is summarized in Algorithm~\ref{alg:broadcast}.

\section{Fault-Tolerant One-to-All Broadcasting}

In this section, we present our contribution to enabling the one-to-all broadcasting algorithm to tolerate node failures.

The key idea is to preserve the original broadcast process by relocating the source node whenever failures disrupt the broadcast propagation. As in standard fault-tolerant broadcasting models, we assume that the original source node \(S\) is operational. Faults are therefore considered among the remaining network nodes.

Faulty nodes can be classified into two categories:
\begin{itemize}
	\item Nodes at graph distance $k$ from the source node $S$,
	\item Nodes at graph distance less than $k$ from the source node.
\end{itemize}

As illustrated in Fig.~\ref{fig:fig3}, nodes located at graph distance \(k\) from the source are leaf-level nodes and therefore do not participate in packet forwarding. Consequently, failures occurring only on the graph-distance-\(k\) boundary do not affect the correctness of the broadcast algorithm. In contrast, failures occurring at graph distances less than \(k\) may interrupt packet propagation and must therefore be explicitly addressed.

If at least one faulty node lies at graph distance less than \(k\), we identify a new source node \(NS\), referred to as the re-rooted source, such that its graph distance from each faulty node is exactly \(k\). Consequently, all faulty nodes behave as leaf-level nodes with respect to the new source and do not interfere with the broadcast process. Once \(NS\) is determined, the packet is first transmitted from the original source \(S\) to \(NS\) using one-to-one routing~\cite{ref3,ref4,ref6}, which requires at most \(k\) steps. Then, \(NS\) initiates the standard one-to-all broadcast.

The faulty nodes do not affect the one-to-one routing step from \(S\) to \(NS\). Since each faulty node is at distance \(k\) from \(NS\), no faulty node can be an internal node on a shortest path from \(S\) to \(NS\). If a faulty node were an internal node on such a path, its remaining distance to \(NS\) would be strictly less than \(k\), contradicting the construction of \(NS\). This argument holds for both single- and two-node failure scenarios.

We now present the algorithms for determining the new source node in the presence of one or two faulty nodes.

\begin{figure}[H]
	\centering
	\begin{subfigure}{0.35\linewidth}
		\centering
		\includegraphics[width=\linewidth]{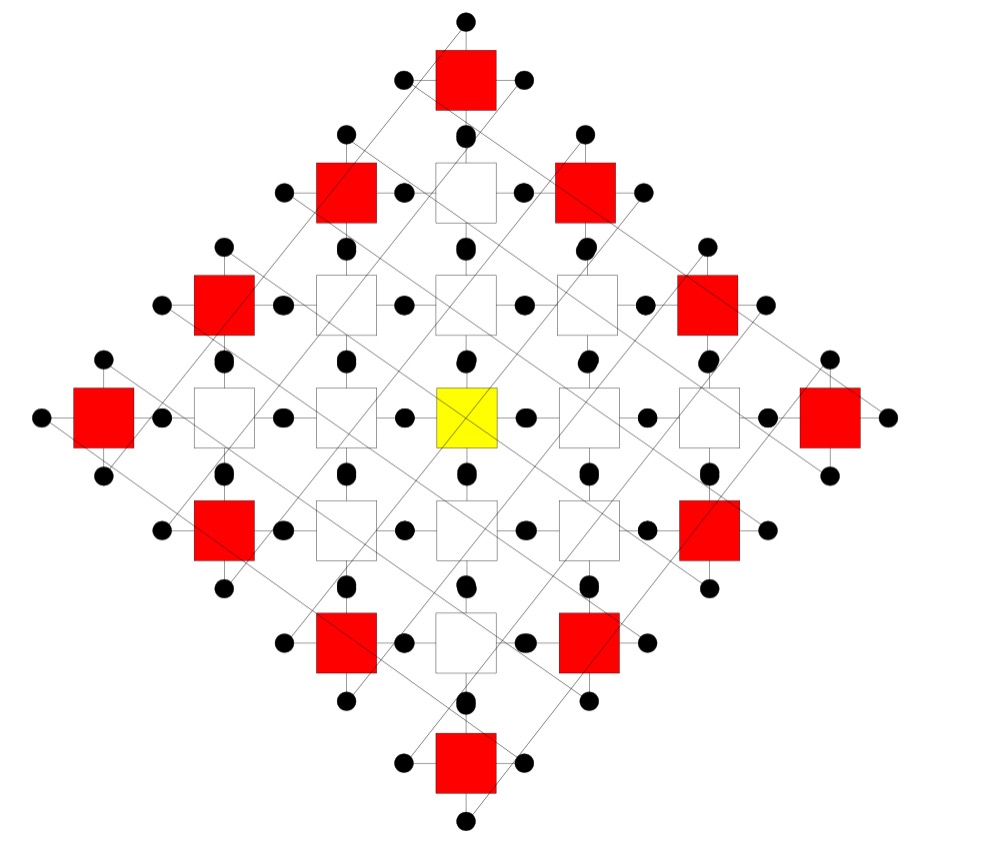}
		\caption{nodes located on the graph-distance-\(k\) boundary from \(S\)}
	\end{subfigure}
	\hfill
	\begin{subfigure}{0.35\linewidth}
		\centering
		\includegraphics[width=\linewidth]{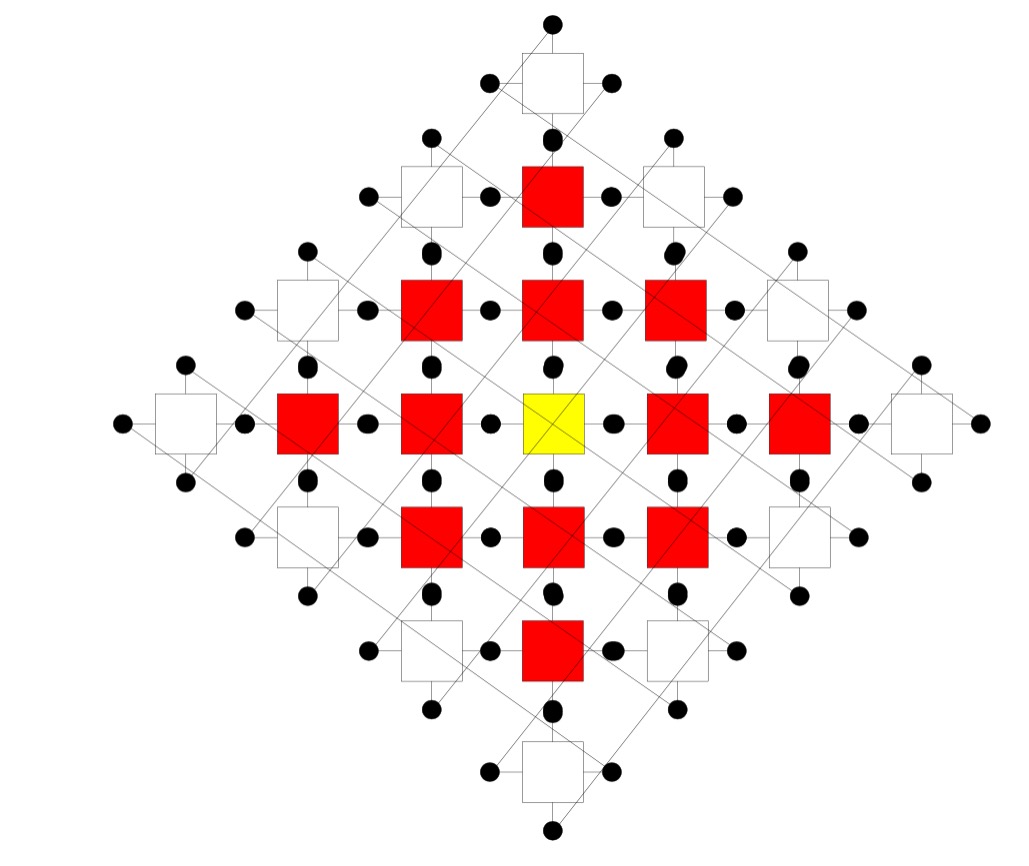}
		\caption{nodes located at graph distance less than \(k\) from \(S\)}
	\end{subfigure}
	\caption{Types of faulty nodes.}
	\label{fig:fig3}
\end{figure}

\subsection{Tolerating One Node Failure}

In the case of a single faulty node \((x,y)\), a new source node can be selected such that its graph distance from \((x,y)\) is exactly \(k\). One simple approach is to choose
\[
NS = ((x,y) + (v,w)) \bmod \pi,
\]
where $(v,w)$ is any node satisfying $|v| + |w| = k$ (for example $(k,0)$).

Alternatively, we provide a more efficient algorithm that avoids modulo operations and relies solely on sums and comparisons. This algorithm iterates over the candidate nodes and returns the first valid solution.

\begin{algorithm}[H]
	\caption{AlgorithmOneFailureFindingNS -- input: $(x,y)$}
	\begin{algorithmic}[1]
		\FOR{$r := -k$ to $k$}
		\STATE $s := k-|r|$
		\IF{$|x+r| + |y+s| \le k$}
		\STATE \textbf{return} $(x+r,y+s)$
		\ENDIF
		\IF{$|x+r| + |y-s| \le k$}
		\STATE \textbf{return} $(x+r,y-s)$
		\ENDIF
		\ENDFOR
	\end{algorithmic}
	\label{alg:onefail}
\end{algorithm}

The complexity of this algorithm is linear in $k$.

\subsection{Tolerating Two Node Failures}

We now consider the case of two faulty nodes $(a,b)$ and $(c,d)$. The goal is to find a node \((r,s)\) whose graph distance from both faulty nodes is equal to \(k\).

To simplify the problem, we first assume that one faulty node is located at $(0,0)$ and solve for an arbitrary node $(x,y)$.  Equivalently, the desired solutions are the nodes in the intersection
\[
B_k(0)\cap \bigl((x,y)+B_k(0)\bigr),
\]
where \(B_k(0)\) denotes the graph-distance-\(k\) boundary around node \(0\).

This problem can be expressed using the following Diophantine equation:

\begin{equation}
	(kx_1+(k+1)y_1) \bmod N = (z+kx_2+(k+1)y_2) \bmod N
\end{equation}
where $z = (kx+(k+1)y) \bmod N$, and $|x_1|+|y_1| = |x_2|+|y_2| = k$.

This equation identifies the intersection of graph-distance-\(k\) boundary nodes centered at \((0,0)\) and \((x,y)\).

The following subsection gives a formal existence proof for the two-fault case.

\subsection{Existence of a Common Distance-$k$ Node}

We now prove that for any two faulty nodes in a dense Gaussian network, there always exists a node whose graph distance from both faulty nodes is exactly \(k\).
For notational simplicity, write
\[
G_k = G(k+(k+1)i),
\]
and recall that
\[
N=k^2+(k+1)^2=2k^2+2k+1.
\]
We use the integer-labeling function \(\phi:G_k\rightarrow \mathbb{Z}_N\) defined in the preliminaries. In particular, for any two nodes \(U,V\in G_k\),
\[
\phi(U-V)\equiv \phi(U)-\phi(V)\pmod N.
\]

\begin{lemma}[Boundary Difference Coverage]
	Let
	\[
	B=\{z\in G_k : d(0,z)=k\}
	\]
	be the set of boundary nodes located at graph distance \(k\) from the origin. Then,
	\[
	\phi(B)-\phi(B)=\mathbb{Z}_N,
	\]
	where subtraction is performed modulo
	\[
	N=2k^2+2k+1.
	\]
\end{lemma}

\begin{proof}
	Consider the boundary side
	\[
	(x,y)=(t,k-t), \qquad 0\le t\le k.
	\]
	These nodes belong to $B$, and their integer labels are
	\[
	\phi(t,k-t)=kt+(k+1)(k-t).
	\]
	Simplifying gives
	\[
	\phi(t,k-t)=k(k+1)-t.
	\]
	Thus, this side produces the consecutive labels
	\[
	C_1=\{k^2,k^2+1,\ldots,k^2+k\}.
	\]
	
	Now consider the opposite boundary side
	\[
	(x,y)=(-t,-(k-t)), \qquad 0\le t\le k.
	\]
	Its labels are
	\[
	\phi(-t,-(k-t))=-kt-(k+1)(k-t).
	\]
	Modulo $N$, these become
	\[
	C_2=\{k^2+k+1,k^2+k+2,\ldots,k^2+2k+1\}.
	\]
	
	Combining both sides gives the consecutive block
	\[
	C=\{k^2,k^2+1,\ldots,k^2+2k+1\}\subseteq \phi(B).
	\]
	
	Next consider the boundary side
	\[
	(x,y)=(-t,k-t), \qquad 0\le t\le k.
	\]
	Its labels are
	\[
	\phi(-t,k-t)=k(k+1)-(2k+1)t.
	\]
	Define
	\[
	P=\{k(k+1)-(2k+1)t : 0\le t\le k\}.
	\]
	Since these labels also belong to boundary nodes,
	\[
	P\subseteq \phi(B).
	\]
	
	Now consider all differences $C-P$. For a fixed $t$, subtract
	\[
	k(k+1)-(2k+1)t
	\]
	from every element of $C$. The resulting interval is
	\[
	[(2k+1)t-k,\,(2k+1)t+k+1].
	\]
	
	For consecutive values of $t$, these intervals touch each other without gaps. The final interval endpoint equals
	\[
	(2k+1)k+k+1=2k^2+2k+1=N.
	\]
	
	Therefore, the union of all intervals covers every residue class modulo $N$. Hence,
	\[
	C-P=\mathbb{Z}_N.
	\]
	
	Since both $C$ and $P$ are subsets of $\phi(B)$, it follows that
	\[
	\phi(B)-\phi(B)=\mathbb{Z}_N.
	\]
\end{proof}

\begin{theorem}
	Let $F_1$ and $F_2$ be any two faulty nodes in the dense Gaussian network $G_k$. Then there exists a node $NS$ such that
	\[
	d(F_1,NS)=k
	\]
	and
	\[
	d(F_2,NS)=k.
	\]
\end{theorem}

\begin{proof}
	Let
	\[
	A\equiv \phi(F_2)-\phi(F_1)\pmod N.
	\]
	
	By the Boundary Difference Coverage Lemma, there exist two boundary nodes $U,V\in B$ such that
	\[
	A\equiv \phi(U)-\phi(V)\pmod N.
	\]
	
	Choose
	\[
	\phi(NS)\equiv \phi(F_1)+\phi(U)\pmod N.
	\]
	
	Since $U\in B$, translation invariance of the Gaussian network implies
	\[
	d(F_1,NS)=k.
	\]
	
	Now,
	\[
	\phi(NS)-\phi(F_2) \equiv \phi(F_1)+\phi(U)-\phi(F_2) \pmod N.
	\]
	
	Using
	\[
	A\equiv \phi(F_2)-\phi(F_1)\pmod N,
	\]
	we obtain
	\[
	\phi(NS)-\phi(F_2) \equiv \phi(U)-A \pmod N.
	\]
	
	Since
	\[
	A\equiv \phi(U)-\phi(V)\pmod N,
	\]
	it follows that
	\[
	\phi(NS)-\phi(F_2) \equiv \phi(V) \pmod N.
	\]
	
	Because $V\in B$,
	\[
	d(F_2,NS)=k.
	\]
	
	Therefore,
	\[
	d(F_1,NS)=d(F_2,NS)=k.
	\]
\end{proof}

The theorem guarantees that a valid re-rooted source node always exists for any pair of faulty nodes in the degree-4 dense Gaussian network. This justifies the proposed two-fault re-rooting strategy. However, the same guarantee cannot be extended to arbitrary triples of faulty nodes, as shown by the following counterexample.

\begin{proposition}
	\label{prop:threefault}
	The two-fault re-rooting guarantee does not generally extend to three faulty nodes.
\end{proposition}

\begin{proof}
	Consider the dense Gaussian network \(G(3+4i)\), where \(k=3\). In this case,
	\[
	N=k^2+(k+1)^2=3^2+4^2=25.
	\]
	Using the integer-labeling function \(\phi\), the graph-distance-\(3\) boundary around node \(0\) is
	\[
	B_3(0)=\{2,5,9,10,11,12,13,14,15,16,20,23\}.
	\]
	This set contains all nodes whose graph distance from node \(0\) is exactly \(3\).
	
	Now choose three faulty nodes with integer labels
	\[
	F_1=0,\qquad F_2=6,\qquad F_3=12.
	\]
	Since the network is vertex-transitive, the graph-distance-\(3\) boundary around a node \(F\) is obtained by shifting \(B_3(0)\) by \(F\) modulo \(25\). Therefore,
	\[
	B_3(F_1)=F_1+B_3(0) =\{2,5,9,10,11,12,13,14,15,16,20,23\},
	\]
	\[
	B_3(F_2)=F_2+B_3(0) =\{1,4,8,11,15,16,17,18,19,20,21,22\},
	\]
	and
	\[
	B_3(F_3)=F_3+B_3(0) =\{0,1,2,3,7,10,14,17,21,22,23,24\}.
	\]
	All additions are performed modulo \(25\). A valid re-rooted source for the three faulty nodes would need to belong to all three graph-distance-\(3\) boundaries. However,
	\[
	B_3(F_1)\cap B_3(F_2)\cap B_3(F_3)=\emptyset.
	\]
	Thus, there is no node \(NS\) satisfying
	\[
	d(NS,F_1)=d(NS,F_2)=d(NS,F_3)=3.
	\]
	Consequently, no common distance-\(k\) re-rooted source exists for this three-fault configuration. Hence, the existence guarantee proved for two faulty nodes cannot be extended to arbitrary triples of faulty nodes.
\end{proof}

This counterexample explains why the present work focuses on one- and two-node fault scenarios. The proposed method is not claimed to solve arbitrary three-fault configurations in degree-4 dense Gaussian networks. Instead, the result establishes a sharp and well-defined reliability guarantee: for any one or two faulty nodes, a valid re-rooted source always exists, whereas for three faulty nodes such a source may not exist.

We now present a simple algorithm that iterates over all nodes on the graph-distance-\(k\) boundary of \((0,0)\) and returns the first node that also satisfies the graph-distance-\(k\) condition with \((x,y)\). To generalize for arbitrary faulty nodes $(a,b)$ and $(c,d)$, we use the coordinate transformation.

\begin{algorithm}[H]
	\caption{AlgorithmAux -- input: $(x,y)$}
	\begin{algorithmic}[1]
		\FOR{$r := -k$ to $k$}
		\STATE $s := k-|r|$
		\IF{$\text{distance}((r,s),(x,y)) = k$}
		\STATE \textbf{return} $(r,s)$
		\ENDIF
		\IF{$\text{distance}((r,-s),(x,y)) = k$}
		\STATE \textbf{return} $(r,-s)$
		\ENDIF
		\ENDFOR
	\end{algorithmic}
	\label{alg:aux}
\end{algorithm}

\begin{algorithm}[H]
	\caption{AlgorithmTwoFailureFindingNS -- input: $(a,b)$, $(c,d)$}
	\begin{algorithmic}[1]
		\STATE $(p,q) := ((a,b)-(c,d)) \bmod \pi$
		\STATE $(r,s) := \text{Aux}(p,q)$
		\STATE \textbf{return} $((r,s)+(c,d)) \bmod \pi$
	\end{algorithmic}
	\label{alg:twofail}
\end{algorithm}

\begin{figure}[H]
	\centering
	\begin{subfigure}{0.35\linewidth}
		\centering
		\includegraphics[width=\linewidth]{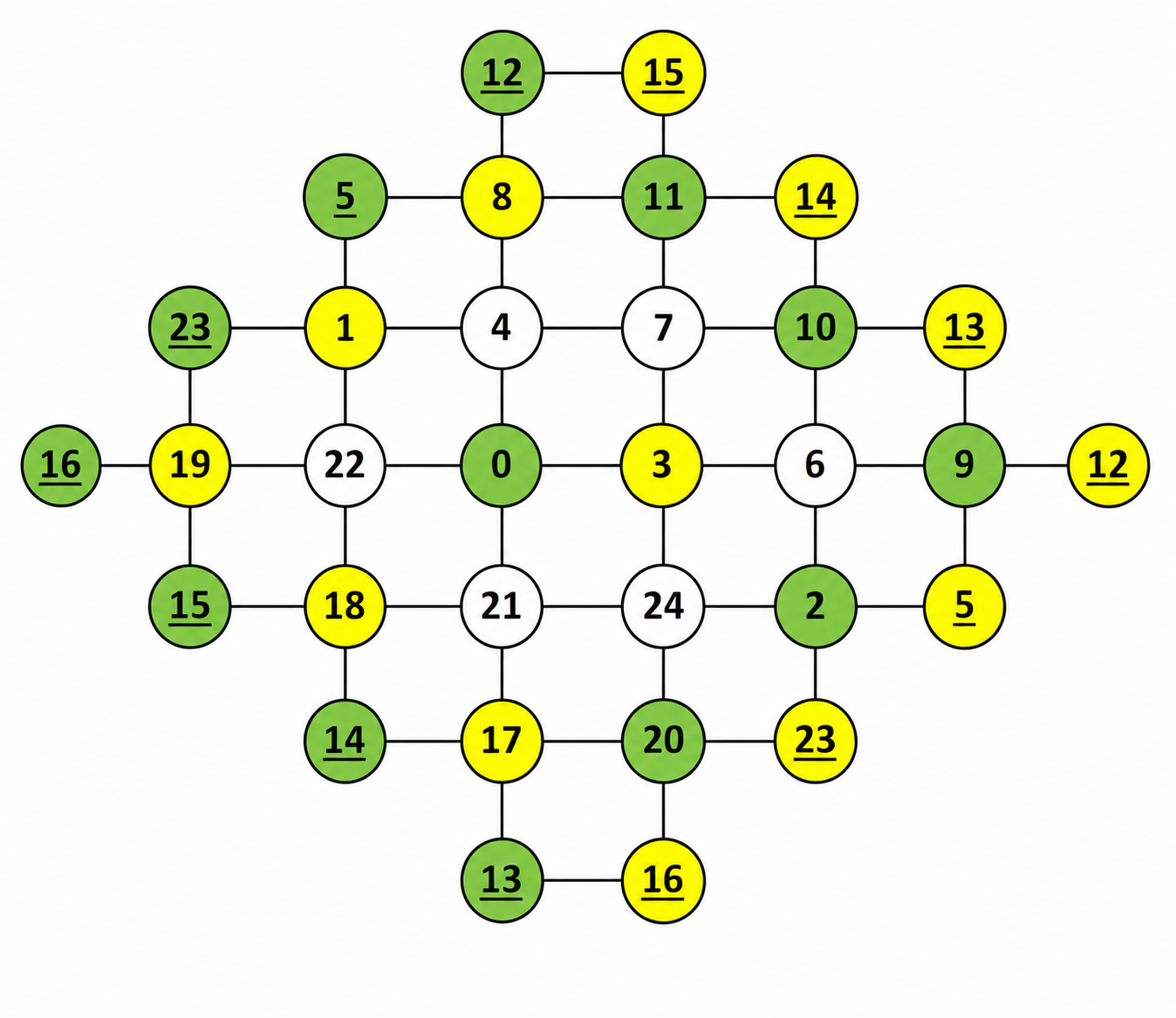}
		\caption{}
	\end{subfigure}
	\hfill
	\begin{subfigure}{0.35\linewidth}
		\centering
		\raisebox{3mm}{
			\includegraphics[width=\linewidth]{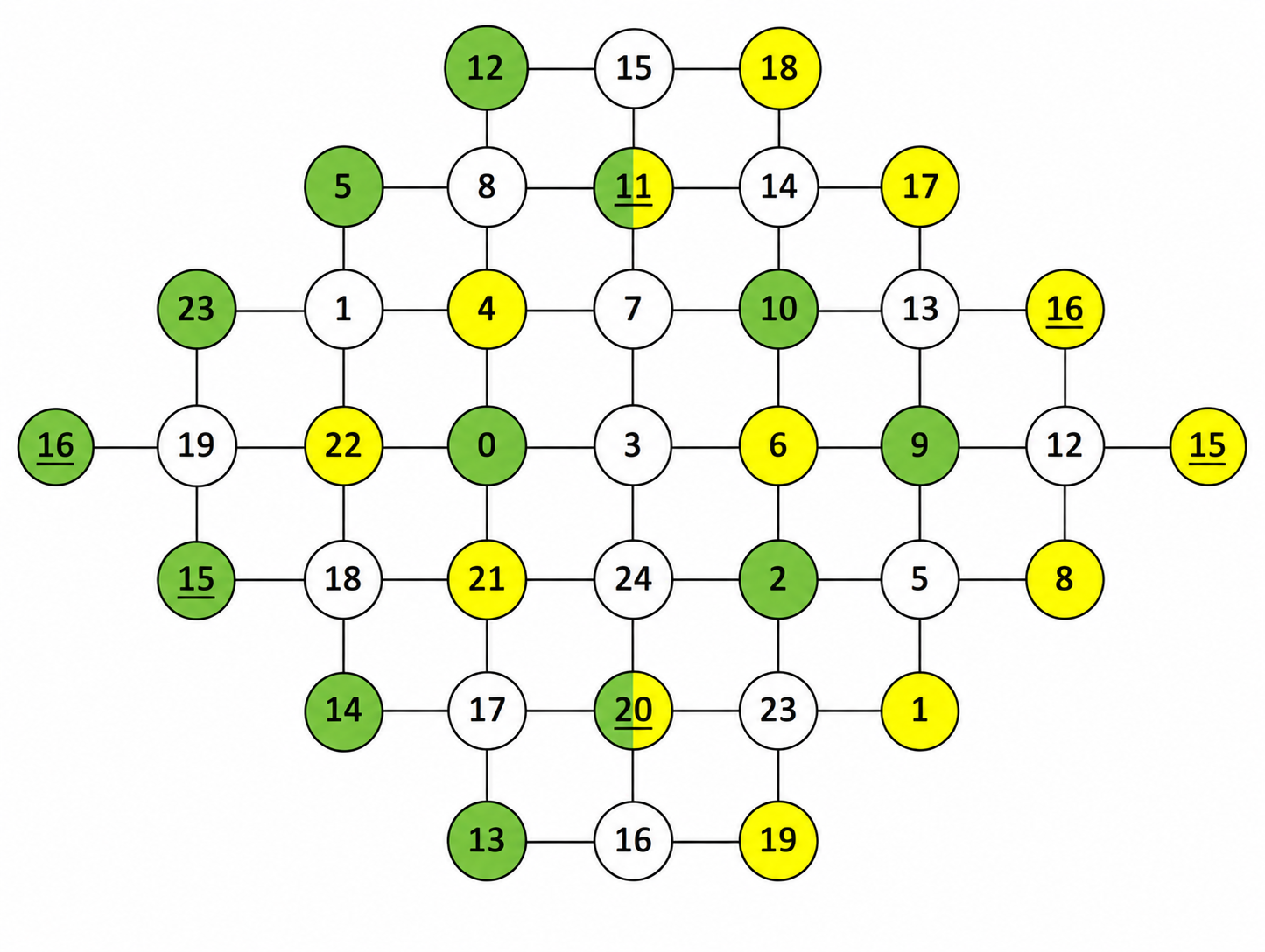}
		}		
		\caption{}
	\end{subfigure}
	\caption{Finding the intersection nodes between nodes at graph distance $k$ from node $0$ and nodes at distance $k$ from (a) node $3$ and (b) node $6$.}
	\label{fig:fig4}
\end{figure}

The transformation procedure described in Algorithm~\ref{alg:twofail} can be understood in three steps and exploits the symmetry of the network.

Fig.~\ref{fig:fig4} illustrates the intersection between graph-distance-\(k\) boundary sets centered at node \(0\) and at another given node. For the \(k=3\) network, Fig.~\ref{fig:fig4}(a) corresponds to the case where the second node is node \(3\), whereas Fig.~\ref{fig:fig4}(b) corresponds to node \(6\). The underlined nodes represent valid solutions whose graph distance from both nodes is exactly \(k\).

For example, in Fig.~\ref{fig:fig4}(b), the second node is node \(6\) in \(G(3+4i)\), and the common graph-distance-\(3\) solutions are nodes \(11\), \(15\), \(16\), and \(20\). The solution node \(11\) satisfies (5) by taking \(x_1=1\), \(y_1=2\), \(x_2=-1\), \(y_2=2\), and \(z=6\). Indeed,
\[
3(1)+4(2)\equiv 6+3(-1)+4(2)\pmod{25}.
\]
Both sides are congruent to \(11\) modulo \(25\). Thus, node \(11\) is a concrete solution belonging to the intersection of the two graph-distance-\(3\) boundary sets.

Fig.~\ref{fig:fig5} shows the number of possible solutions for each node in networks with $k=4$ and $k=5$. The value shown in brackets represents the number of nodes satisfying the graph-distance-$k$ condition from both node $0$ and the corresponding node. The green cells represent nodes with $2k+1$ solutions, the yellow cells represent nodes with four solutions, and the white cells represent intermediate cases with a number of solutions between $4$ and $2k$.

\begin{figure}[H]
	\centering
	\begin{subfigure}{0.45\linewidth}
		\centering
		\includegraphics[width=\linewidth]{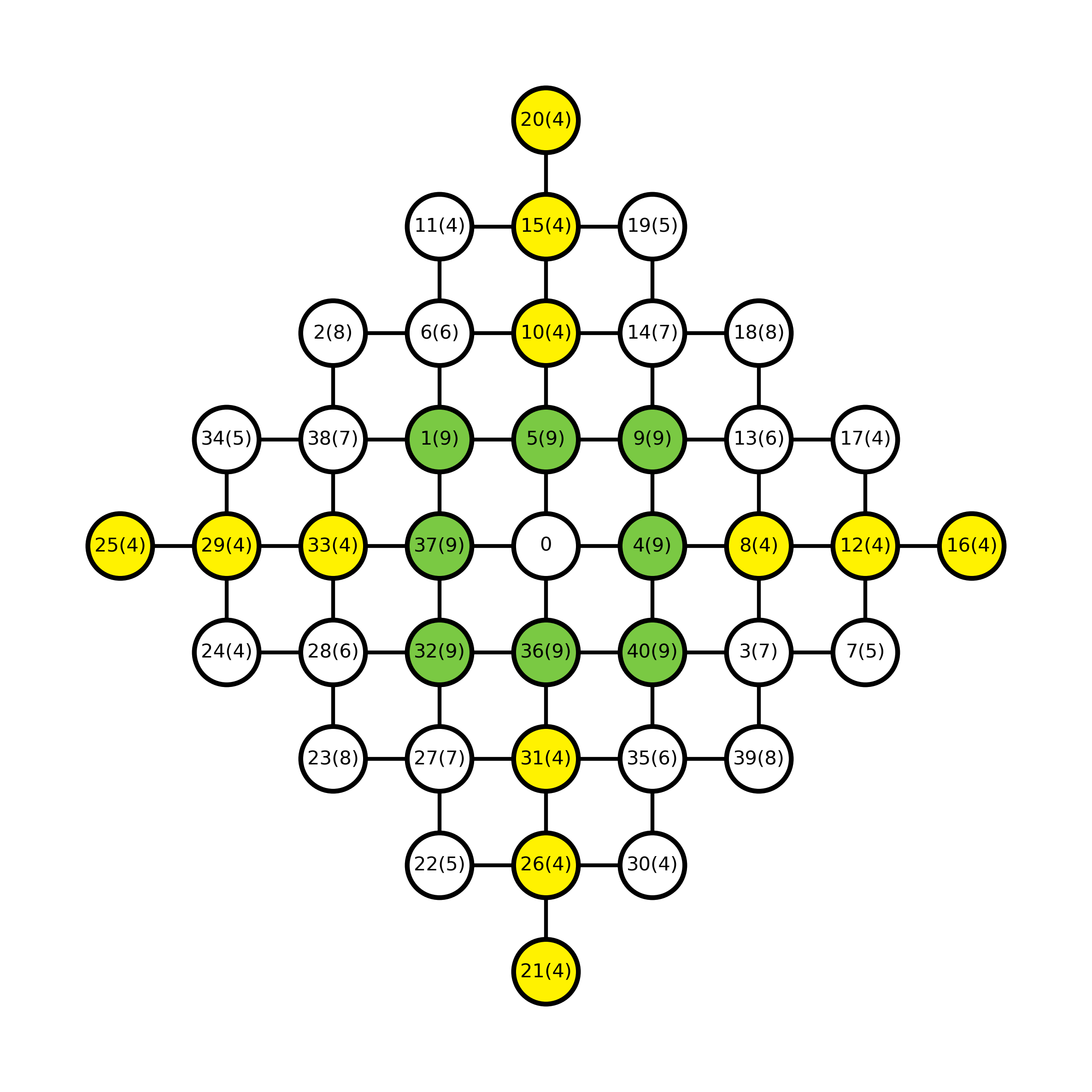}
		\caption{$k=4$}
	\end{subfigure}
	\hfill
	\begin{subfigure}{0.45\linewidth}
		\centering
		\includegraphics[width=\linewidth]{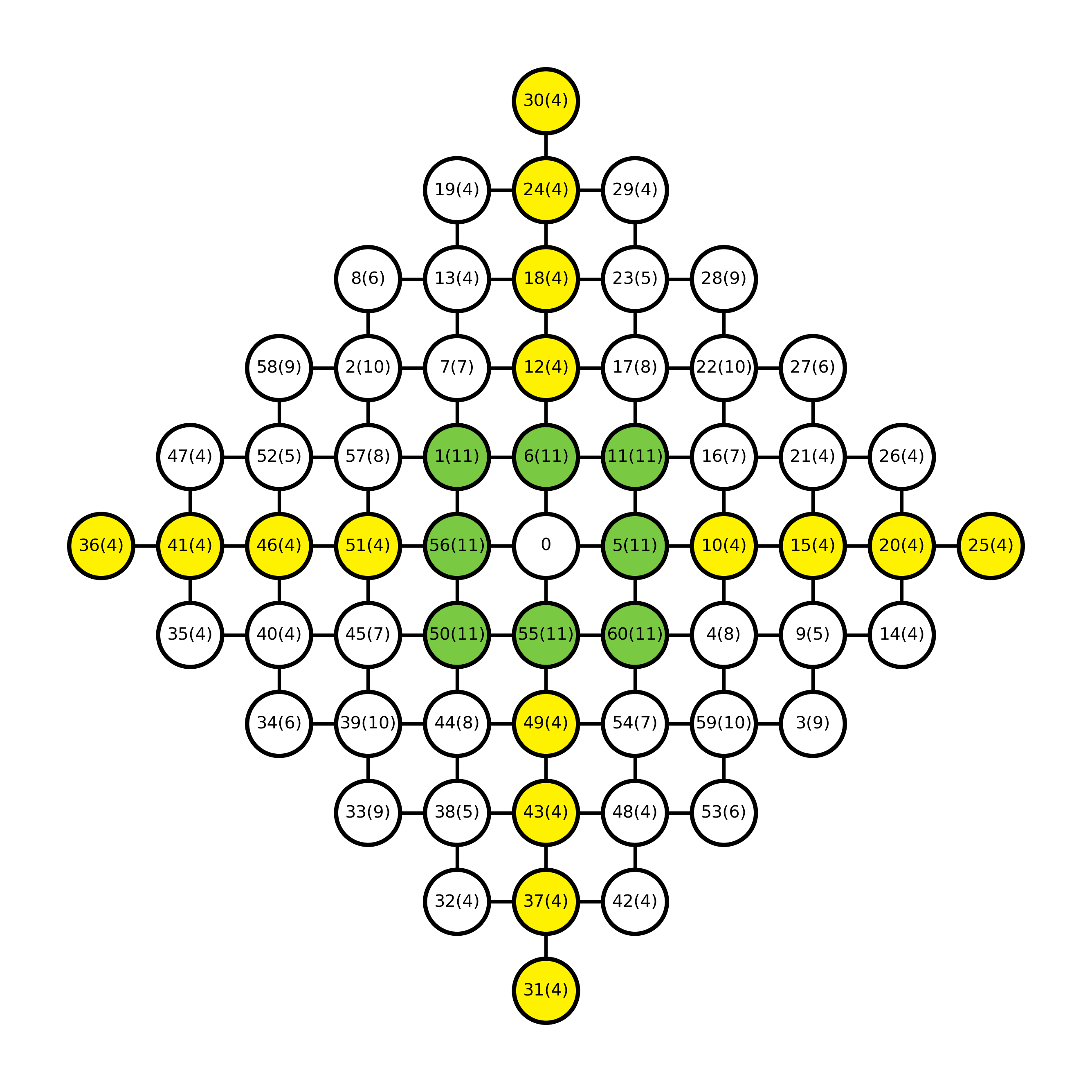}
		\caption{$k=5$}
	\end{subfigure}
	\caption{Number of solutions for two faulty nodes.}
	\label{fig:fig5}
\end{figure}

\textbf{Example:} Assume that the two faulty nodes are $(0,2)$ and $(-1,-1)$ in a $k=3$ network. The equivalent circulant node numbers are $8$ and $18$, respectively. Using the vertex symmetry of the network, the first step computes the relative position of the first node with respect to the second node:
\[
(8-18) \bmod 25 = -10 \bmod 25 = 15.
\]
In the Gaussian integer representation, this corresponds to $(-2,-1)$. Therefore, the auxiliary algorithm is called with $(-2,-1)$ as input.

The auxiliary algorithm searches the nodes at distance $k$ from $(0,0)$ and returns the first node that is also at distance $k$ from $(-2,-1)$. In this case, the returned node is $(0,-3)$.

Finally, the coordinate transformation is reversed:
\begin{align*}
	((0,-3)+(-1,-1)) \bmod (3,4) \\
	= (-1,-4) \bmod (3,4) \\
	= (2,0).
\end{align*}
Thus, $(2,0)$ is selected as the new source node NS.

\subsection{Complexity and Communication Cost Analysis}

The proposed fault-tolerant broadcasting method consists of two phases:
\begin{enumerate}
	\item selecting a new source node \(NS\), and
	\item transmitting the message from the original source \(S\) to \(NS\), followed by the standard one-to-all broadcast from \(NS\).
\end{enumerate}

For the single-fault case, Algorithm~\ref{alg:onefail} scans candidate nodes on the graph-distance-\(k\) boundary. Since the loop ranges over a linear number of boundary candidates, the number of inspected positions is \(O(k)\). Therefore, the source-selection complexity is
\[
T_1(k)=O(k).
\]

For the two-fault case, Algorithm~\ref{alg:twofail} first applies a constant-time coordinate transformation and then calls Algorithm~\ref{alg:aux}. The auxiliary algorithm scans candidate positions on the graph-distance-\(k\) boundary, whose size is linear in \(k\). Hence, the worst-case source-selection complexity is
\[
T_2(k)=O(k).
\]
Thus, the proposed source-selection procedure remains linear in the network diameter.

The standard one-to-all broadcasting algorithm in the dense Gaussian network completes in \(k\) parallel communication steps. This is because the broadcast expands simultaneously through the four directional regions of the network. Therefore, in the absence of faults, the broadcast time is
\[
B_{\text{baseline}}(k)=k.
\]

The proposed re-rooting method adds only one preliminary one-to-one relocation phase from the original source \(S\) to the selected new source \(NS\). Since the diameter of \(G(k+(k+1)i)\) is \(k\), the relocation distance satisfies
\[
d(S,NS)\leq k.
\]
Thus, the additional relocation time is bounded by
\[
R(k)\leq k.
\]

After the message reaches \(NS\), the same standard one-to-all broadcasting algorithm is executed without modification. Therefore, the total worst-case broadcast time of the proposed method is
\[
B_{\text{total}}(k) = R(k)+B_{\text{baseline}}(k).
\]
Since \(R(k)\leq k\) and \(B_{\text{baseline}}(k)=k\), we obtain
\[
B_{\text{total}}(k)\leq 2k.
\]

Thus, the proposed method preserves the linear-time broadcast behavior of the original algorithm. The fault-free broadcast completes in \(k\) parallel steps, while the fault-tolerant re-rooted broadcast completes in at most \(2k\) steps in the worst case. The additional delay is only the one-to-one relocation step, and this delay is bounded by the network diameter.

In terms of computation, the new source node is found in \(O(k)\) time for both one- and two-fault cases. In terms of communication time, the additional overhead is at most \(k\) steps. Therefore, the proposed method achieves one- and two-node fault tolerance without constructing redundant spanning trees, backup broadcast structures, or multiple disjoint routing paths.

\section{Experimental Evaluation}

In this section, we evaluate the performance and scalability of the proposed re-rooting-based fault-tolerant broadcasting approach through simulations. The goal is to compare the reliability and efficiency of the proposed method with those of the baseline one-to-all broadcasting algorithm described in Section~IV. The evaluation considers both one-node and two-node failure scenarios across multiple network sizes and fault-placement modes.

\subsection{Experimental Setup}

We conducted simulations on dense Gaussian networks of varying sizes corresponding to
\[
k = 10,25,50,100,200,
\]
resulting in networks with
\[
N=k^2+(k+1)^2
\]
nodes. Thus, the tested network sizes were \(221\), \(1301\), \(5101\), \(20201\), and \(80401\) nodes, respectively. These sizes allow us to evaluate the scalability of the proposed re-rooting strategy from small networks to large dense Gaussian networks.

For each network size, experiments were conducted under two node-failure scenarios:
\begin{itemize}
	\item \textbf{Single-node failure}: one faulty node is present in the network.
	\item \textbf{Two-node failures}: two distinct faulty nodes are present in the network.
\end{itemize}

For each value of \(k\) and each fault scenario, we tested four fault-placement modes:
\begin{itemize}
	\item \textbf{Random}: faulty nodes are selected uniformly at random.
	\item \textbf{Near-source}: faulty nodes are selected close to the original source.
	\item \textbf{Critical-position}: faulty nodes are selected from positions likely to interfere with broadcast forwarding.
	\item \textbf{Close-pair}: faulty nodes are selected from locally clustered regions; in the two-fault case, this corresponds to selecting faulty nodes with small mutual graph distance.
\end{itemize}

For each fault-placement mode, \(250\) independent trials were performed. Therefore, each \(k\)-and-fault configuration contains \(1000\) trials in total, aggregated across the four modes. Since five values of \(k\) and two fault scenarios were tested, the full experiment contains \(10000\) trials.

We evaluated two broadcasting approaches:
\begin{itemize}
	\item \textbf{Baseline broadcasting}: the original one-to-all broadcasting algorithm without fault tolerance.
	\item \textbf{Proposed re-rooting method}: the dynamic source-relocation strategy followed by the standard one-to-all broadcasting algorithm.
\end{itemize}

The following performance metrics were measured:
\begin{itemize}
	\item \textbf{Broadcast success rate}: the percentage of trials in which all non-faulty nodes receive the broadcast message.
	\item \textbf{Reachability}: the average number of non-faulty nodes reached by the broadcast.
	\item \textbf{Reachability standard deviation}: the standard deviation of the number of reached non-faulty nodes across trials.
	\item \textbf{Relocation distance}: the number of hops required to transmit the message from the original source to the selected new source.
	\item \textbf{Total broadcast steps}: the sum of the relocation distance and the \(k\)-step one-to-all broadcast from the new source.
	\item \textbf{Runtime}: the computation time required to identify the new source node.
\end{itemize}

For each metric, we report the average value over the aggregated trials. The standard deviation is reported separately to show the stability of the baseline and proposed methods under different fault placements.

\subsection{Simulation Procedure and Reproducibility}

For each tested value of \(k\), the simulator first constructs the dense Gaussian network
\[
G(k+(k+1)i)
\]
with
\[
N=k^2+(k+1)^2
\]
nodes. Each node is represented using its integer label in \(\mathbb{Z}_N\), and adjacency is generated according to the dense circulant representation \(C_N(k,k+1)\). Thus, each node \(u\) is connected to
\[
u\pm k \pmod N \quad \text{and} \quad u\pm(k+1)\pmod N.
\]

For each trial, the faulty node set is generated according to the selected fault-placement mode. In the single-fault case, one node is marked as faulty. In the two-fault case, two distinct nodes are marked as faulty. The original source node is fixed as node \(0\), and the same set of faulty nodes is used to evaluate both the baseline and proposed methods.

The baseline method executes the original one-to-all broadcasting algorithm directly from the original source. During propagation, a faulty node can receive the message but cannot forward it to other nodes. Therefore, if a faulty node lies on an internal forwarding position, the broadcast may fail to reach all non-faulty nodes.

The proposed method first computes a new source node \(NS\). For one faulty node, \(NS\) is selected at graph distance \(k\) from the faulty node. For two faulty nodes, \(NS\) is selected so that its graph distance from both faulty nodes is equal to \(k\). The message is then transmitted from the original source to \(NS\), after which the original one-to-all broadcasting algorithm is executed from \(NS\).

A trial is considered successful if every non-faulty node receives the broadcast message. Formally, if \(F\) is the set of faulty nodes and \(R\) is the set of nodes reached by the broadcast, then the trial is successful when
\[
G_k \setminus F \subseteq R.
\]
Equivalently, the number of reached non-faulty nodes must be
\[
N-|F|.
\]

The same success criterion is used for both the baseline and proposed methods. This allows the comparison to isolate the effect of re-rooting while keeping the underlying broadcast procedure unchanged.

\subsection{Results and Analysis}

Table~\ref{tab:results} summarizes the broadcast success rate and average reachability of the baseline and proposed methods. Table~\ref{tab:statistics} reports the standard deviation of the number of reached nodes across the fault-placement experiments. Table~\ref{tab:overhead} reports the communication and computational overhead of the proposed method, including relocation distance, total broadcast steps, and source-selection runtime.

The results show that the proposed re-rooting method achieved \(100\%\) success in every tested configuration. In all cases, it reached exactly \(N-|F|\) nodes, where \(|F|\) is the number of faulty nodes. In contrast, the baseline broadcast success rate decreased as the network size increased and was especially low under two-node failures. This behavior occurs because the baseline method fails whenever a faulty node occupies an internal forwarding position in the original broadcast process.

\begin{table}[H]
	\centering
	\caption{Broadcast performance under one- and two-node failures across all fault-placement modes.}
	\label{tab:results}
	\begin{adjustbox}{max width=\textwidth}
	\begin{tabular}{c c c c c c c c}
		\hline
		\(k\) & Failures & \(N\) & Trials & Baseline Success (\%) & Proposed Success (\%) & Avg. Baseline Reach & Avg. Proposed Reach \\
		\hline
		10  & 1 & 221   & 1000 & 10.2 & 100.0 & 202.243   & 220 \\
		25  & 1 & 1301  & 1000 & 4.9  & 100.0 & 1207.433  & 1300 \\
		50  & 1 & 5101  & 1000 & 2.7  & 100.0 & 4759.311  & 5100 \\
		100 & 1 & 20201 & 1000 & 1.1  & 100.0 & 18929.323 & 20200 \\
		200 & 1 & 80401 & 1000 & 0.6  & 100.0 & 75099.441 & 80400 \\
		\hline
		10  & 2 & 221   & 1000 & 5.3  & 100.0 & 189.966   & 219 \\
		25  & 2 & 1301  & 1000 & 1.6  & 100.0 & 1145.970  & 1299 \\
		50  & 2 & 5101  & 1000 & 1.4  & 100.0 & 4519.830  & 5099 \\
		100 & 2 & 20201 & 1000 & 0.3  & 100.0 & 18090.168 & 20199 \\
		200 & 2 & 80401 & 1000 & 0.3  & 100.0 & 71815.562 & 80399 \\
		\hline
	\end{tabular}
	\end{adjustbox}
\end{table}

\begin{figure}[H]
	\centering
	\includegraphics[width=0.65\linewidth]{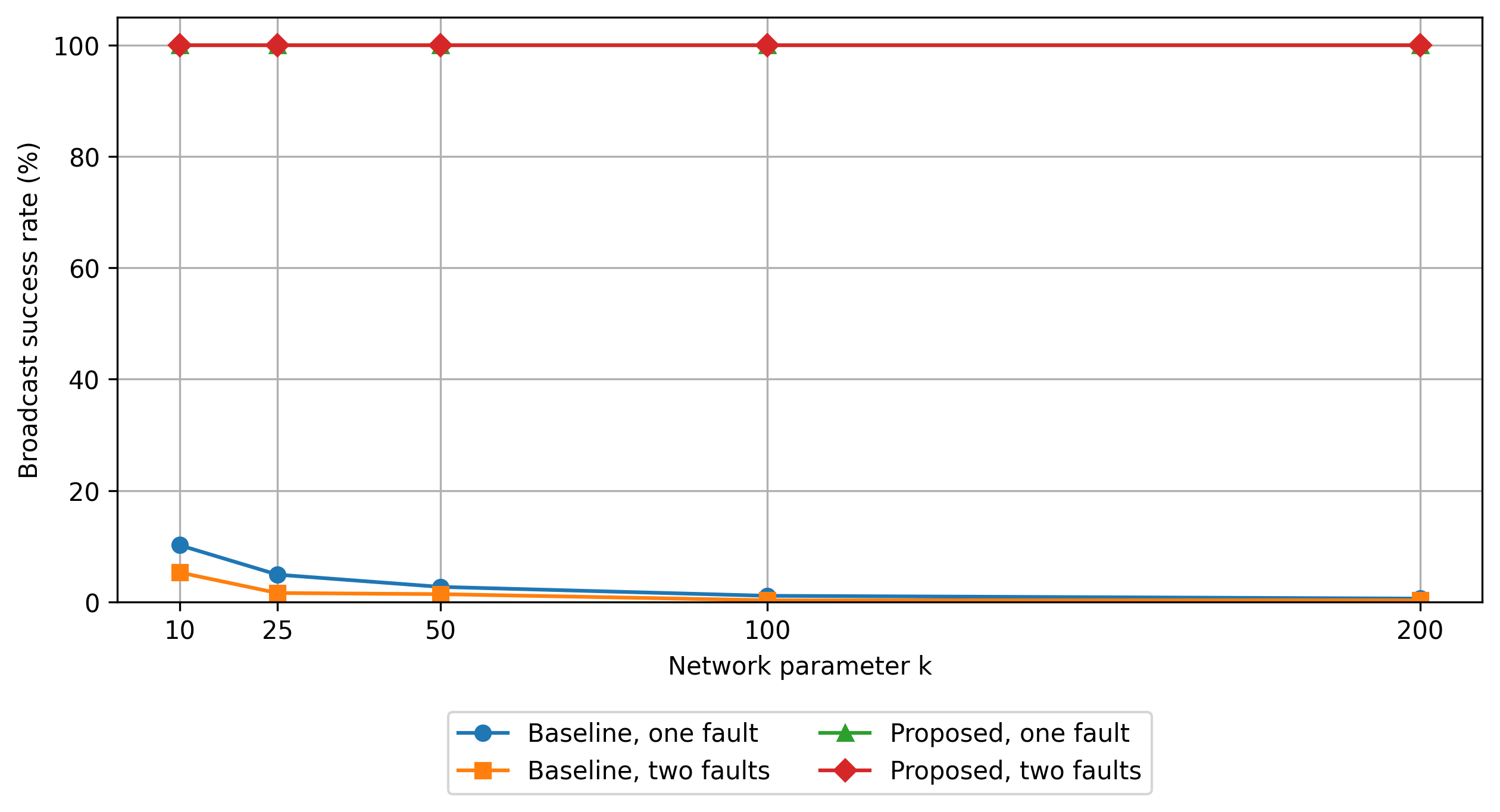}
	\caption{Broadcast success rate versus network parameter \(k\). The proposed re-rooting method maintains \(100\%\) success across all tested network sizes and fault counts, while the baseline success rate decreases as \(k\) increases.}
	\label{fig:success_scalability}
\end{figure}

\begin{table}[t]
	\centering
	\caption{Statistical summary of fault-placement experiments.}
	\label{tab:statistics}
	\begin{adjustbox}{max width=\textwidth}
	\begin{tabular}{c c c c}
		\hline
		\(k\) & Faults & Baseline reach std. dev. & Proposed reach std. dev. \\
		\hline
		10  & 1 & 19.793528    & 0.00 \\
		25  & 1 & 123.037375   & 0.00 \\
		50  & 1 & 489.391615   & 0.00 \\
		100 & 1 & 1941.363752  & 0.00 \\
		200 & 1 & 7940.349833  & 0.00 \\
		\hline
		10  & 2 & 27.497035    & 0.00 \\
		25  & 2 & 176.232139   & 0.00 \\
		50  & 2 & 756.117317   & 0.00 \\
		100 & 2 & 2810.814650  & 0.00 \\
		200 & 2 & 11550.529034 & 0.00 \\
		\hline
	\end{tabular}
\end{adjustbox}
\end{table}

Because each aggregated \(k\)-and-fault configuration contains \(1000\) independent trials, the reported success rates provide stable empirical estimates across the tested fault-placement modes. The proposed method achieved successful delivery in all \(1000\) trials for every configuration, yielding zero observed failures. Using the rule of three for zero-failure binomial observations, the empirical failure probability of the proposed method is below approximately \(0.3\%\) at the \(95\%\) confidence level for each tested configuration. For the baseline method, the variation in reachability is reported using standard deviation, while the success-rate trends provide an empirical indication of how sensitive the original broadcast process is to the placement of faulty nodes.

\begin{table}[H]
\centering
\caption{Communication-step and computational overhead of the proposed re-rooting method.}
\label{tab:overhead}
\begin{adjustbox}{max width=\textwidth}
\begin{tabular}{c c c c c c c}
	\hline
	\(k\) & Failures & \(N\) & Avg. Relocation Hops & Avg. Total Steps & Min--Max Total Steps & Avg. Runtime (ms) \\
	\hline
	10  & 1 & 221   & 8.731   & 18.731  & 10--20   & 0.423555 \\
	25  & 1 & 1301  & 22.651  & 47.651  & 26--50   & 2.572521 \\
	50  & 1 & 5101  & 46.460  & 96.460  & 50--100  & 10.207384 \\
	100 & 1 & 20201 & 91.937  & 191.937 & 103--200 & 42.117149 \\
	200 & 1 & 80401 & 183.832 & 383.832 & 200--400 & 182.397966 \\
	\hline
	10  & 2 & 221   & 7.462   & 17.462  & 10--20   & 0.405332 \\
	25  & 2 & 1301  & 18.585  & 43.585  & 26--50   & 2.698035 \\
	50  & 2 & 5101  & 37.582  & 87.582  & 51--100  & 11.627536 \\
	100 & 2 & 20201 & 74.549  & 174.549 & 102--200 & 42.826660 \\
	200 & 2 & 80401 & 151.404 & 351.404 & 206--400 & 178.875955 \\
	\hline
\end{tabular}
\end{adjustbox}
\end{table}

\begin{figure}[H]
\centering
\includegraphics[width=0.65\linewidth]{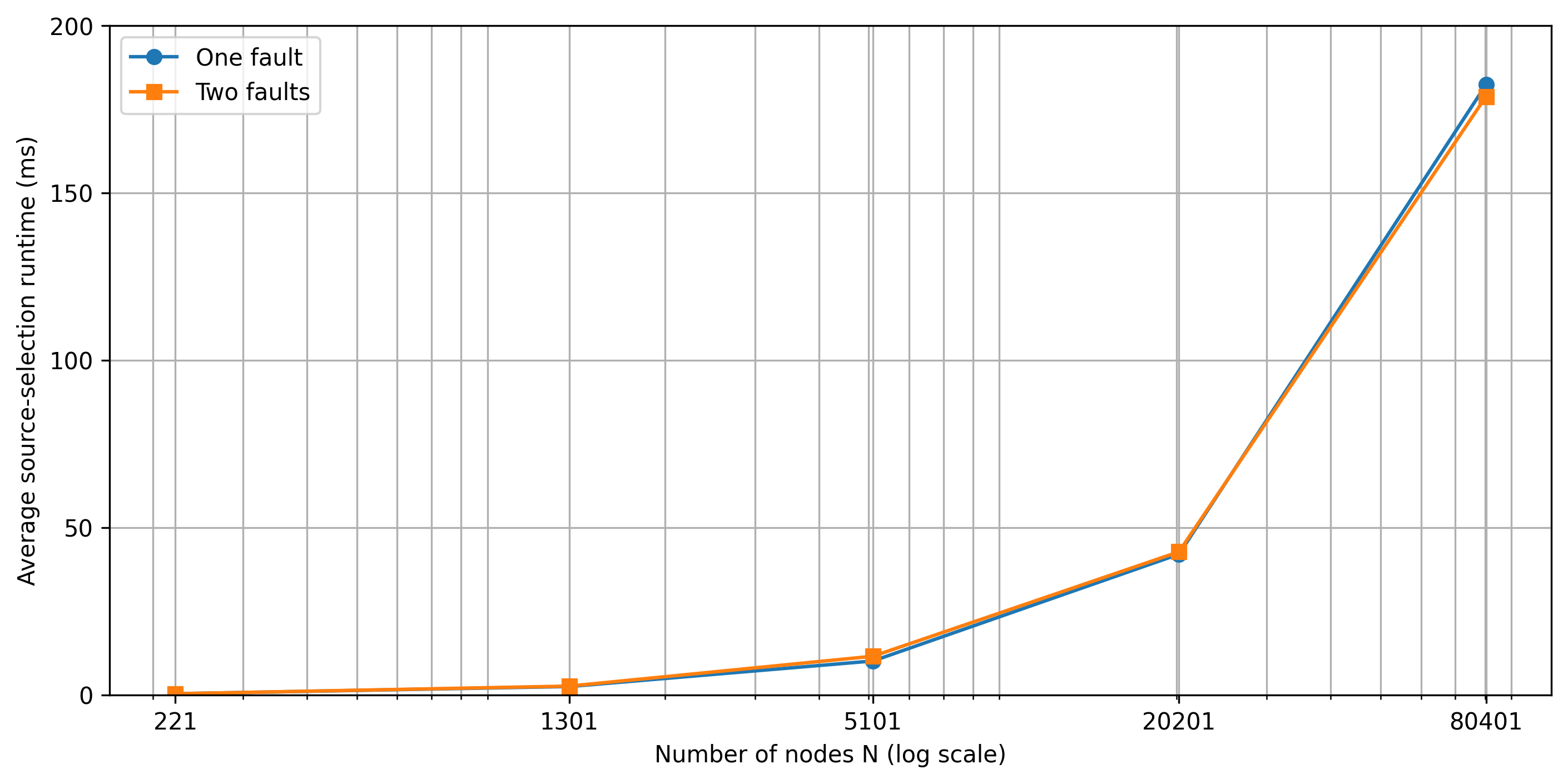}
\caption{Average source-selection runtime versus network size using a logarithmic \(N\)-axis. The proposed method remains computationally lightweight even for the largest tested network with \(80401\) nodes.}
\label{fig:runtime_scalability}
\end{figure}

Figures~\ref{fig:success_scalability} and~\ref{fig:runtime_scalability} illustrate the scalability of the proposed method. Figure~\ref{fig:success_scalability} shows that the proposed re-rooting method maintains complete delivery across all tested values of \(k\), whereas the baseline success rate decreases as the network grows. Figure~\ref{fig:runtime_scalability} shows that the runtime required to identify the new source increases with network size but remains practical even for \(N=80401\). These results support the analytical complexity bound and demonstrate that the proposed method scales to networks containing tens of thousands of nodes.

The statistical results show that the proposed method has zero reachability variation because it reached exactly \(N-|F|\) nodes in every trial. In contrast, the baseline method shows increasing variation as \(k\) grows, reflecting its sensitivity to the location of faulty nodes within the original broadcast structure.

The overhead results show that the proposed method introduces only one preliminary one-to-one relocation phase from the original source to the selected new source. Since the dense Gaussian network has diameter \(k\), the relocation distance is at most \(k\). Therefore, the total broadcast time is bounded by \(2k\): at most \(k\) steps for relocation and exactly \(k\) parallel steps for the one-to-all broadcast from \(NS\). The maximum total-step values in Table~\ref{tab:overhead} confirm this bound for all tested values of \(k\).

The results demonstrate the effectiveness of the proposed approach. While the baseline algorithm suffers significant degradation in the presence of node failures, the proposed method consistently achieves complete delivery across all tested network sizes, failure counts, and fault-placement modes. This confirms the theoretical guarantee that placing faulty nodes at graph distance \(k\) from the new source makes them leaf-level nodes in the broadcast process.

In the baseline approach, even a single faulty node located within an internal forwarding position can disrupt broadcast propagation and lead to incomplete coverage. This effect becomes more pronounced as the network size increases because there are more possible internal forwarding positions where a failure can interrupt message propagation. In contrast, the proposed method neutralizes the impact of the faulty nodes by relocating the broadcast source before the one-to-all broadcast begins.

Overall, the proposed approach achieves a favorable trade-off between reliability and efficiency. It provides deterministic one- and two-node fault tolerance for dense Gaussian networks while preserving the original \(k\)-step parallel broadcasting procedure and adding only a bounded relocation phase.

\subsection{Comparative Discussion with Fault-Tolerant Broadcasting Approaches}

\begin{table}[H]
\centering
\caption{Qualitative comparison with common fault-tolerant broadcasting strategies.}
\label{tab:qualitative_comparison}
\begin{adjustbox}{max width=\textwidth}
\begin{tabular}{l c c c c}
	\hline
	Method & Extra structures & Backup paths & Runtime adaptation & Broadcast algorithm changed \\
	\hline
	Redundant spanning trees & Yes & Yes & Low/Medium & Yes \\
	Multiple-path routing & Yes & Yes & Medium & Yes \\
	Adaptive rerouting & No/Partial & Partial & High & Yes \\
	Local recovery & Partial & Partial & Medium & Yes \\
	Proposed re-rooting & No & No & Low & No \\
	\hline
\end{tabular}
\end{adjustbox}
\end{table}

Although the experimental comparison in this work focuses on the original non-fault-tolerant broadcasting algorithm, it is also useful to position the proposed method relative to common fault-tolerant broadcasting strategies. Existing approaches for fault-tolerant broadcast in interconnection networks often rely on redundant spanning trees, multiple backup paths, adaptive rerouting, or local recovery mechanisms. These techniques improve reliability by adding alternative communication structures or by dynamically avoiding failed components during message propagation.

The proposed re-rooting method follows a different design philosophy. Instead of modifying the broadcast tree or constructing redundant delivery paths, it changes the effective source of the broadcast. The new source \(NS\) is selected so that the faulty nodes lie on the graph-distance-\(k\) boundary of the network with respect to \(NS\). As a result, the faulty nodes become leaf-level nodes in the broadcast process and are not required to forward the message to other nodes.

Table~\ref{tab:qualitative_comparison} summarizes the main qualitative differences between the proposed method and representative fault-tolerant broadcasting strategies.

The main advantage of the proposed approach is its simplicity. It does not require precomputed backup trees, disjoint path construction, or modification of the original one-to-all broadcasting procedure. Once \(NS\) is selected, the same broadcasting algorithm is executed from the new source. Therefore, the additional cost is limited to the source-selection step and the one-to-one relocation of the message from the original source to \(NS\).

However, this comparison is qualitative rather than a full empirical benchmark against all existing fault-tolerant broadcasting families. A complete implementation-level comparison with independent spanning tree methods, adaptive rerouting schemes, and redundant-path algorithms would require selecting specific representative algorithms and adapting them to the same dense Gaussian network model. Such an extended benchmark is beyond the scope of this paper, but the qualitative comparison clarifies the main trade-off: the proposed method offers a lightweight and topology-specific solution for one- and two-node failures, while more general fault-tolerant methods may support broader failure models at the cost of additional routing structures or higher runtime complexity.

\subsection{Failure Model and Limitations}

The proposed re-rooting method is designed for static node-failure scenarios. In this model, the set of faulty nodes is known before the broadcast begins and does not change during message propagation. Under this assumption, the new source node \(NS\) can be selected so that the faulty nodes are located at graph distance \(k\) from \(NS\), making them leaf-level nodes in the broadcast process.

This work focuses on one- and two-node failures. As proved earlier, for any pair of faulty nodes in a degree-4 dense Gaussian network, there always exists a node \(NS\) whose graph distance from both faulty nodes is exactly \(k\). Therefore, the proposed method provides a deterministic guarantee for all one- and two-node fault configurations. However, this guarantee does not extend to arbitrary three-node failures, as demonstrated by the counterexample in Proposition~\ref{prop:threefault}. For this reason, the proposed method is intentionally limited to one- and two-node failure scenarios.

The present model does not directly address link failures. A failed link may interrupt communication even if both endpoint nodes remain operational. In such a case, placing a faulty node on the graph-distance-\(k\) boundary is not sufficient, because the failure affects an edge rather than a forwarding node. Extending the re-rooting idea to link failures would require a different condition, possibly based on avoiding failed edges along the directional broadcast paths rather than selecting a source that makes faulty nodes leaves.

The model also assumes that failures are static during the broadcast. If a new node fails after the message has already been relocated to \(NS\), the selected source may no longer satisfy the required distance condition. Handling dynamic or transient failures would require either repeated re-rooting, online failure detection, or adaptive recovery during message propagation. These extensions are outside the scope of the present work.

Thus, the contribution of this paper is a deterministic and lightweight re-rooting method for static one- and two-node failures in dense Gaussian networks. Broader failure models, including link failures, transient faults, and dynamically appearing faults, are important directions for future work.

\subsection{Practical Implications}

The proposed re-rooting approach is particularly suitable for systems in which the underlying communication topology is regular and the broadcast algorithm is already optimized for parallel propagation. In such environments, adding redundant spanning trees or maintaining several backup routes may increase hardware complexity, routing-table storage, control overhead, and verification cost. By contrast, the proposed method preserves the original one-to-all broadcast algorithm and adds only a lightweight source-selection and relocation phase.

This property is attractive for NoC-based chip multiprocessors, many-core accelerators, distributed edge-computing fabrics, and embedded parallel systems, where predictable communication behavior and low implementation overhead are important. The method is also compatible with hardware-friendly routing because the selected new source is determined using the structure of the dense Gaussian network rather than by continuously adapting the broadcast path during message propagation.

The results suggest that re-rooting can be used as a topology-specific reliability mechanism when the expected failure model is limited to a small number of static node failures. In this setting, the method offers a practical trade-off: it improves broadcast reliability without requiring redundant broadcast structures, while keeping the worst-case communication time bounded by \(2k\).

The simulation procedure, network construction method, failure modes, and evaluation metrics are described to support reproducibility of the reported results.

\section{Conclusion}

This paper presented a lightweight re-rooting-based fault-tolerant broadcasting method for dense Gaussian networks. The main idea is to relocate the effective broadcast source so that faulty nodes are placed at graph distance \(k\), the network diameter, from the new source. Under the standard one-to-all broadcasting process, such nodes become leaf-level nodes and are therefore not required to forward the broadcast message.

The proposed method was developed for both single-node and two-node failure scenarios. For the single-fault case, a suitable new source can be selected directly from the graph-distance-\(k\) boundary of the faulty node. For the two-fault case, we proved that for any pair of faulty nodes in \(G(k+(k+1)i)\), there always exists a node whose graph distance from both faulty nodes is exactly \(k\). This establishes a deterministic existence guarantee for the proposed two-fault re-rooting strategy.

The paper also identified the limitation of this guarantee. A counterexample in \(G(3+4i)\) shows that a common graph-distance-\(k\) node does not necessarily exist for arbitrary triples of faulty nodes. Therefore, the method is intentionally positioned as a deterministic solution for one- and two-node failures in degree-4 dense Gaussian networks, rather than as a general solution for all multi-fault configurations.

The cost analysis shows that the source-selection procedure requires \(O(k)\) computation. The original one-to-all broadcast completes in \(k\) parallel communication steps, and the re-rooting method adds only one preliminary one-to-one relocation step whose length is at most \(k\). Thus, the total worst-case broadcast time is bounded by \(2k\) steps, while preserving the original broadcasting algorithm and avoiding redundant spanning trees, backup paths, or additional broadcast structures.

Simulation results confirm that the proposed method achieves complete delivery to all non-faulty nodes in the tested one- and two-node failure scenarios, while the baseline broadcast may fail when faulty nodes occur at internal forwarding positions. These results support the theoretical analysis and demonstrate that source relocation can provide an efficient and topology-specific mechanism for improving broadcast reliability in dense Gaussian networks.

Future work will investigate extensions of the re-rooting idea to broader failure models, including link failures, dynamic or transient faults, and higher-order node failures. Another important direction is to compare re-rooting experimentally with specific independent-spanning-tree and adaptive-routing methods under a unified dense Gaussian network simulation framework.

\section*{Acknowledgments}
The authors would like to acknowledge the support of Kuwait University and its Computer Science Department.

\end{document}